\documentclass[12pt,draftcls,onecolumn]{IEEEtran}
\usepackage{amsthm,dsfont}
\usepackage{algorithmic,algorithm}
\usepackage{supertabular}
\usepackage{cite}
\usepackage{color}
\usepackage{graphics} 
\usepackage{latexsym, amsmath, amsfonts, amssymb,graphicx,array,tabularx}
\usepackage{subfigure}
\usepackage{caption}
\usepackage[marginal]{footmisc}
\usepackage{epstopdf}
\theoremstyle{theorem}

\newtheorem{theorem}{Theorem}[section]
\newtheorem{lemma}{Lemma}

\newtheorem{remark}{Remark}

\newtheorem{definition}{Definition}
\newtheorem{assumption}{Assumption}
\newtheorem{problem}{Problem}

\makeatletter

\newcommand{\Rmnum}[1]{\expandafter\@slowromancap\romannumeral #1@}
\makeatother

\newcommand\addtag{\refstepcounter{equation}\tag{\theequation}}
\DeclareMathOperator*{\esssup}{ess\,sup}
\DeclareMathOperator*{\argmax}{arg\,max}

\newcommand{\smallO}[1]{o{\left(#1\right)}}

\begin{document}
\IEEEoverridecommandlockouts

\title{Quickest Change Detection with a Censoring Sensor in the Minimax Setting}

\author{Xiaoqiang $\mathrm{Ren}^{\star}$, Jiming $\mathrm{Chen}^{\ddag}$, Karl H. $\mathrm{Johansson}^{\dag}$ and Ling $\mathrm{Shi}^{\star}$

\thanks{\indent $\star$: Department of Electronic and Computer Engineering,
Hong Kong University of Science and Technology, Kowloon, Hong Kong. Emails: xren,~eesling@ust.hk.

        \indent $\ddag$: Institute of Industrial Process Control, Department of Control Science and Engineering, Zhejiang University, Hangzhou, China. E-mail: jmchen@iipc.zju.edu.cn.

        \indent $\dag$: ACCESS Linnaeus Center, School of Electrical Engineering, Royal Institute of Technology, Stockholm, Sweden. Email: kallej@ee.kth.se.
}}
\maketitle

\begin{abstract}
The problem of quickest change detection with a wireless sensor node is studied in this paper. The sensor that is deployed to monitor the environment has limited energy constraint to the classical quickest change detection problem. We consider the ``censoring" strategy at the sensor side, i.e., the sensor selectively sends its observations to the decision maker. The quickest change detection problem is formulated in a minimax way. In particular, our goal is to find the optimal censoring strategy and stopping time such that the detection delay is minimized subject to constraints on both average run length (ARL) and average energy cost before the change. We show that the censoring strategy that has the maximal post-censoring Kullback-Leibler (K-L) divergence coupled with Cumulative Sum (CuSum) and Shiryaev-Roberts-Pollak (SRP) detection procedure is asymptotically optimal for the Lorden's and Pollak's problem as the ARL goes to infinity, respectively. We also show that the asymptotically optimal censoring strategy should use up the available energy and has a very special structure, i.e., the likelihood ratio of the \emph{no send} region is a single interval, which can be utilized to significantly reduce the computational complexity. Numerical examples are shown to illustrate our results.

\end{abstract}

\section{Introduction}
The problem of quickest change detection aims to detect an abrupt change in stochastic processes as soon as possible. This arises in a wide range of applications including quality control engineering \cite{lai1995sequential}, finance \cite{shiryaev2002quickest}, computer security \cite{thottan2003anomaly, cardenas2009evaluation}, condition monitoring \cite{rice2010flexible} and cognitive radio networks \cite{jayaprakasam2009sequential}. In these applications, a change of the underlying distribution usually indicates that the event we are interested in occurs, and in order to take actions, we need to detect such occurrence as soon as possible. Generally speaking, the design of quickest change detection procedures, mainly involves two performance indices: detection delay and false detection. Usually, one seeks to find the detection procedure that minimizes detection delay subject to certain false detection constraint. The mathematical characterization of these two indices and the model assumption distinguishes two problem formulations: the Bayesian one due to Shiryaev \cite{Shiryaev1963, shiryaev2007optimal} and the minimax one due to Lorden \cite{lorden1971procedures} and Pollak \cite{pollak1985optimal}. The strict optimal or asymptotically optimal detection procedures have been established for these problems; see surveys \cite{polunchenko2012state, veeravalli2012quickest}.

Quickest change detection with wireless sensor networks (WSNs) has raised recent interest \cite{Veeravalli2001, mei2011quickest, banerjee2012dataBayesian, banerjee2012data, Geng2013}. The reasons are twofold: one is due to the good inherent properties the WSNs possess, such as being flexible and robust, and the other one is that the quickest change detection procedure can be carried out with WSNs in many applications including infrastructure health monitoring, habit monitoring and intrusion detection. The limited resources (limited energy for battery-powered sensor nodes and limited bandwidth for communication channels) associated with WSNs, however, brings new challenges. In the setting of classical quickest change detection, it is assumed that the decision maker can access the observation at each time instant and the sampling cost is free, whereas one has to take the energy and bandwidth constraints into consideration for detection with WSNs and the decision maker usually can only access part of the observation. In summary, the quickest change detection with WSNs needs to deal with trade-offs among three performance indices: detection delay, false detection and energy cost.

There have been several proposals on the quickest change detection with energy constraint; see \cite{premkumar2008optimal, banerjee2012dataBayesian, banerjee2012data}. The existing approaches admit the following two features. One is that the energy constraint is characterized by the sampling cost, i.e., the number of observations made. To meet this energy constraint, the observations are taken only at certain time instants. The other feature is that the decision whether or not to take a sample is made at the fusion center based on the detection statistic.
To cope with the energy constraint, we consider ``censoring" strategy at the sensor nodes instead. Specifically, the sensor nodes take observations at each time instant, but only send those that are deemed as ``informative" to the decision center. The benefits to adopt censoring strategy at the sensor nodes are summarized as follows.
On one hand, in most practical applications, the energy consumption of sensing is negligible compared with that of communication \cite{dargie2010fundamentals}, so it is effective to reduce the total energy consumption by reducing the number of communications. On the other hand, if the decision whether or not to take observations is made based on the detection statistic at the fusion center, the feedback information from the center to the sensor nodes will be needed, which can cause additional energy consumption and bandwidth cost.

In this paper, the minimax problem formulation of quickest change detection with a censoring sensor node is studied. Like the classical minimax formulation, we consider both Lorden's and Pollak's problem and try to find the optimal censoring strategy and stopping time such that the detection delay is minimized subject to constraints on both average run length (ARL) and average energy cost before the change. The main contributions of our work are summarized as follows.
\begin{enumerate}
  \item To the best of our knowledge, this paper is the first work that studies the minimax formulation of quickest change detection with a sensor that adopts censoring strategy. It is shown in our numerical example that the censoring strategy provides a very good trade-off between the detection performance and the energy constraint.
  \item We show that the censoring strategy that has the maximal post-censoring Kullback-Leibler (K-L) divergence coupled with Cumulative Sum (CuSum) and Shiryaev-Roberts-Pollak (SRP) detection procedure is asymptotically optimal for both Lorden's and Pollak's problem as the ARL goes to infinity, respectively. (Theorem \ref{Theorem:AsymOptCuSum} and Theorem \ref{Theorem:AsymOptSRP})
  \item In general, to find the asymptotically optimal censoring strategy that maximizes the post-censoring K-L divergence can only be done numerically. The computation burden of searching over the whole admissible class is huge, especially when the dimension of the observation is high. To alleviate the computation burden, we provide two necessary conditions on the asymptotically optimal censoring strategy. One is that it should use up the available energy (Theorem \ref{Theorem:EqualGreaterKL}), the other is that it has a very special structure, i.e., the likelihood ratio of the \emph{no send} region is a single interval (Theorem \ref{Theorem:likelihoodratio}).

\end{enumerate}


The related literature are summarized as follows. The idea of detection with censoring strategy is introduced in \cite{rago1996censoring} and later studied in \cite{appadwedula2008decentralized} and \cite{Tay2007}. The main result of \cite{rago1996censoring} and \cite{appadwedula2008decentralized} is that the likelihood ratio of the censoring region is a single interval under several different performance indices. The asymptotic detection performance of large-scale censoring networks is studied in \cite{Tay2007}. Premkumar and Kumar \cite{premkumar2008optimal} considered Bayesian quickest change detection with sleep/awake control of the sensor nodes. The energy constraint is formulated as the average number of sensors used and the problem is solved by formulating it as an infinite-horizon Markovian decision process problem. Banerjee and Veeravalli \cite{banerjee2012dataBayesian} studied a similar problem but with one sensor node. The authors provided an asymptotically optimal low-complexity stopping rule. The same authors studied the minimax problem with the same energy formulation in \cite{banerjee2012data}. They proposed a heuristical ``DE-CuSum" (``data efficient" CuSum) algorithm and proved that the algorithm is asymptotically optimal. Geng and Lai \cite{Geng2013} studied minimax quickest change detection with a sensor that can be recharged with energy harvesting techniques. The authors proposed a very simple asymptotically optimal power allocation scheme.

The remainder of this paper is organized as follows. The mathematical model of the considered problem is given in Section \ref{Section:problem setup}. In Section \ref{Section:Main results}, we show the main results of this paper. First we prove that the asymptotically optimal censoring strategy for both Lorden's and Pollak's problem is the one that has the maximal post-censoring K-L divergence. Then two properties of the asymptotically optimal censoring strategy are shown, i.e., that it uses up the available energy and the likelihood ratio of the \emph{no send} region is a single interval. Numerical examples are given in Section \ref{Section:Numerical Example} to illustrate the main results. Some concluding remarks are presented in the end.

\textit{Notations}: $\mathbb{N}$, $\mathbb{N}_{+}$, $\mathbb{R}$, $\mathbb{R}_{+}$ and $\mathbb{R}_{++}$ are the set of non-negative integers, positive integers, real numbers, non-negative real numbers and positive real numbers, respectively. $k\in \mathbb{N}$ is the time index. 
$\mathbf{1}_{A}$ represents the indicator function that takes value $1$ on the set $A$ and $0$ otherwise.
 $\times$ stands for the Cartesian product and $\text{Pr}(\cdot)$ denotes the probability.

\section{Problem Setup} \label{Section:problem setup}
We consider the optimal censoring strategy and detection scheme for a wireless sensor system that adopts censoring strategy. A remote sensor is deployed to take observations from the environment at each time instant and selectively sends them to a center that makes decisions to continue or declare a change of the monitored environment sequentially. By ``selectively", we mean that due to limited resources (e.g., energy and bandwidth), the sensor cannot send its observations all the time.
To make the most of the limited energy and achieve better detection performance, it is natural to come up with a strategy that only sends ``informative" data and discards the ``less informative" ones, which is the concept of censoring.
For detection structure, we consider the scenario corresponding to case A in \cite{veeravalli2001decentralized}, where the sensor has no local memory and does not receive feedback from the center, i.e., the sensor makes decisions whether or not to communication with the center only based on its current observation; see Fig. \ref{Fig:topology}. The reasons why we consider the scenario where there is only one sensor being used are as follows. On one hand, if we assume the observations are independent across the different sensors, the multiple sensors case can be easily reduced to the one sensor case for the problem we study, so for brevity of presentation, the one sensor case is preferred. On the other hand, if the observations taken from different sensors may be correlated, the problem of quickest change detection with energy constraints will be much more complicated, which we shall leave to the future work. In fact, most related literature studied the one sensor case, such as \cite{banerjee2012dataBayesian, banerjee2012data, Geng2013, geng2013bayesian }.  In the following, we introduce the mathematical problem.

\begin{figure}
  \centering
  \includegraphics[width=2.5in]{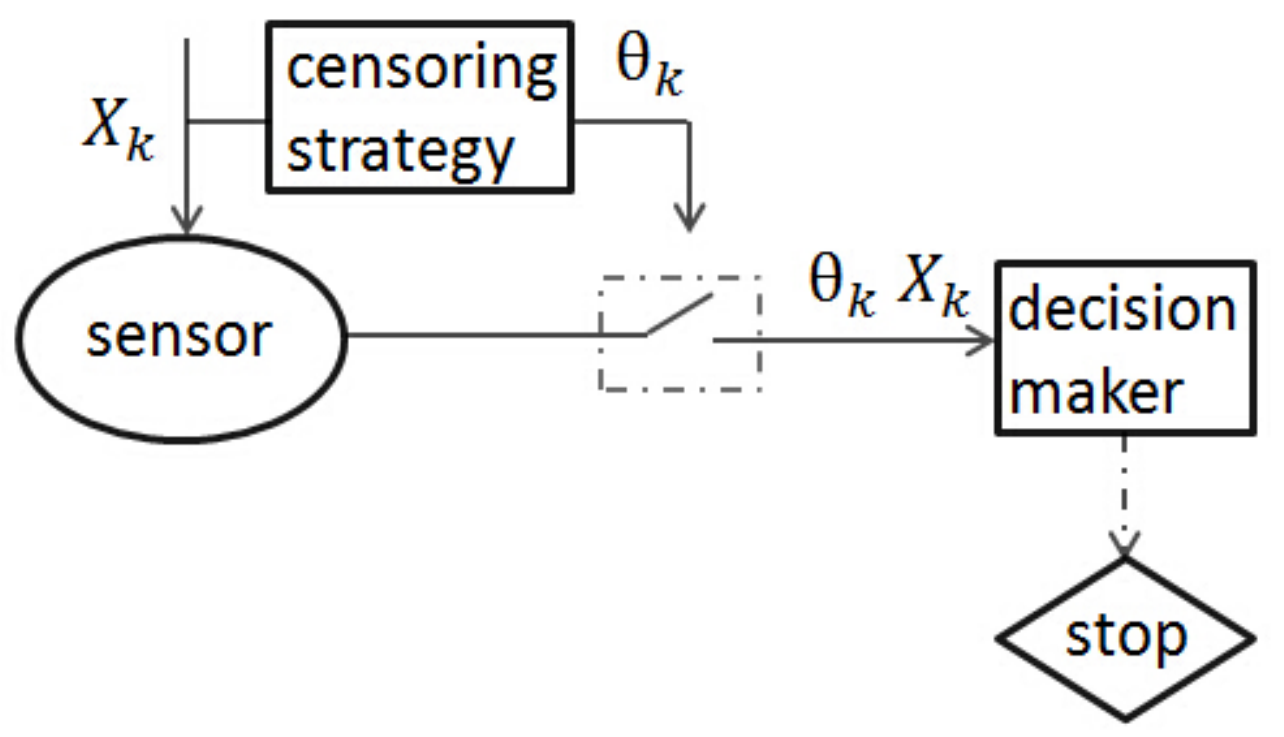}\vspace{-2mm}
  \caption{ Block diagram of the quickest change detection system. } \label{Fig:topology}
  \vspace{-4mm}
\end{figure}

\subsection{Classical Minimax Quickest Change Detection}
Let $X_k$ be the observation at time $k$. Along the time horizon, a sequence of observations $\{X_k\}_{k\in\mathbb{N}_{+}}$ about the monitored environment are taken locally at the sensor. Let $\{X_k\}_{k\in\mathbb{N}_{+}}$ be a series of random variables defined on the probability space $(\Omega, \mathcal{F}, \mathbb{P})$.   Assume that before an unknown but not random time instant $\nu$ ($\nu$ may be $\infty$ and in this case the change never happens), the environment is ``in control", and the interested change event happens at $\nu$, after which the system is ``out of control" and some measures must be taken as soon as possible. Specifically, the observations at the sensor before $\nu$, $X_1,X_2,\ldots,X_{\nu-1}$ denoted by $X_1^{\nu-1}$ are i.i.d. with measure $\mathbb{P}_{\infty}$, and $X_{\nu}^{\infty}$ are i.i.d. with measure $\mathbb{P}_{1}$.
Note that $\mathbb{P}_{\nu}$ denote the probability measure when the change happens at $\nu$. If there is no change, we denote this measure by $\mathbb{P}_{\infty}$. The expectation $\mathbb{E}_{\nu}$ and $\mathbb{E}_{\infty}$ are defined accordingly. We assume $\mathbb{P}_{\infty}$ and $\mathbb{P}_1$ are mutually locally absolutely continuous measures.

Define the filtration $\{\mathcal{F}_k\}_{k\in\mathbb{N}}$ induced by the observations $\{X_k\}_{k\in\mathbb{N}_{+}}$ as
\[\mathcal{F}_k\triangleq\sigma(X_1,X_2,\ldots,X_k), \forall k\in\mathbb{N}_{+}, \]
with the understanding that $\mathcal{F}_0$ is the trivial $\sigma$-algebra. Note that $\mathcal{F}=\vee_{k\geq 0} \mathcal{F}_k$. A random variable $T\in\mathbb{N}_{+}$ is called a stopping time if the event $\{T=k\}\in \mathcal{F}_k, \forall k\in \mathbb{N}_{+}$. The objective of classical quickest change detection is to find a stopping time $T$ that detects the change event \emph{as quickly as possible}, subject to the \emph{risk of false detection} \cite{polunchenko2012state}. For the different problem formulations (e.g., Bayesian and minimax), there are different criteria how to measure ``\emph{as quickly as possible}" and ``\emph{risk of false detection}".

In the minimax formulation of the quickest change detection, one aims to minimize the worst detection delay in a sense. In the minimax setting, the risk of false detection is measured by the average run length (ARL) to false alarm $\mathbb{E}_{\infty}[T]$ \cite{lorden1971procedures}.
As the detection delay is concerned, there are mainly two criteria: Lorden's ``worst-worst case" detection delay \cite{lorden1971procedures} and Pollak's ``worst case" conditional average delay \cite{pollak1985optimal}. Given a stopping time $T$, the associated Lorden's detection delay is defined by
\begin{align}
\mathcal{D}_L(T) \triangleq \sup_{0\leq\nu<\infty} \big\{ \esssup \mathbb{E}_{\nu}[(T-\nu)^+|\mathcal{F}_{\nu}] \big\},
\end{align}
where $(T-\nu)^+ = (T-\nu)\mathbf{1}_{\{(T-\nu)\geq0\}}$. Pollak's detection delay is defined by
\begin{align}
\mathcal{D}_P(T) \triangleq \sup_{0\leq\nu<\infty} \mathbb{E}_{\nu}[(T-\nu)^+|T>\nu].
\end{align}

Denote by $\mathcal{C}_{\gamma}$ the class of stopping times that has the ARL lower bounded by $\gamma \geq 1$, i.e.,
\[\mathcal{C}_{\gamma} \triangleq \{T:\mathbb{E}_{\infty}[T]\geq \gamma\}.\]
Lorden's minimax optimization problem is to find the optimal stopping time $T^*$ such that
\begin{align}
\mathcal{D}_L(T^*)=\inf_{T\in\mathcal{C}_{\gamma}}\mathcal{D}_L(T), \:\text{for every}\: \gamma \geq 1.
\end{align}
Similarly, the Pollak's minimax optimization problem aims to find the optimal stopping time  $T^*$ such that
\begin{align}
\mathcal{D}_P(T^*)=\inf_{T\in\mathcal{C}_{\gamma}}\mathcal{D}_P(T), \:\text{for every}\: \gamma \geq 1.
\end{align}
Lorden \cite{lorden1971procedures} proved that Page's CuSum detection procedure \cite{page1954continuous} is first-order asymptotically optimal as $\gamma \to \infty$. Later, Moustakides \cite{moustakides1986optimal} and Ritov \cite{ritov1990decision} proved that the CuSum procedure is in fact strictly optimal for any $\gamma>1$. The strict optimal detection procedure for the Pollak's problem is still an open problem. In the regime of asymptotic optimality, it has been proved that the CuSum algorithm is first-order asymptotically optimal \cite{lai1998information} and some variants of the Shiyaev-Roberts (SR) procedure (e.g., the SR-$r$ and SRP procedure) are third-order asymptotically optimal \cite{tartakovsky2012third}.

\subsection{Minimax Quickest Change Detection with Energy Constraint}
The classical minimax quickest change detection formulations, however, do not consider the cost of sending observations. This is, however, an important issue when a wireless sensor is used.  In this paper, we take the energy constraint into account and study the variation of the classical minimax quickest change detection.

We consider the censoring strategy at the sensor's side. After taking an observation at each time instant, the sensor needs to decide whether or not to send it to the decision maker. Assume that each communication between the sensor and the center will cost the sensor $1$ unit of energy. Let the binary-valued variable $\mu_k$ be the censoring rule at time $k$ and $\vec{\mu}=\{\mu_1,\ldots,\mu_k,\ldots\}$ be the decision policy.
Note that since it is assumed that there is no feedback or memory for the sensor, decision is made based only on current observations,
i.e.,
\[\theta_k=\mu_k(X_k).\]
We consider stationary policies of the form $\vec{\mu}=\{\mu,\ldots,\mu,\ldots\}$. The reasons why we consider stationary policy are as follows. First, the stationary policy can facilitate local processing and strategy implementation at the sensor nodes. Second, compared with time-varying policies, the stationary policy is easier to be analyzed, which can provide insights of the benefits of censoring strategies for the quickest change detection. Third, a reasonable time-varying policy is the one where the censoring strategy varies with the detection statistic. Then the observations available at the center are correlated. We leave such complicated case to the future work. As $\vec{\mu}$ is fully specified by $\mu$, in the sequel we will use these two notations interchangeably.

We pose the energy constraint by
\begin{align}
\mathbb{E}_{\infty}[\mu(X_1)] \leq \epsilon, \label{Eqn:energyconstraint}
\end{align}
where
$0 < \epsilon \leq 1$ is the upper bound of the average units of energy used at each time before the change happens.
We ignore the energy cost constraint after the change event as the cost of the detection system still working after the change is already penalized by the detection delay. The parameter $\epsilon$ can be tuned to achieve desired trade-off between energy cost and detection performance.
Note that since the censoring strategy is stationary and $X_1^{\infty}$ are i.i.d under $\mathbb{P}_{\infty}$, the above energy constraint can be rewritten as
\begin{align*}
\mathbb{E}_{\infty}[\mu(X_k)] \leq \epsilon, \, \forall k\geq 1.
\end{align*}
Such an energy constraint naturally arises in wireless sensor networks and similar ones have been studied in \cite{banerjee2012data}.
In particular, the above energy constraint is equivalent to the pre-change duty cycle (PDC) introduced in \cite{banerjee2012data}. Like PDC, the energy constraint defined in \eqref{Eqn:energyconstraint} can be rewritten by
\begin{align}
\limsup_{n}\frac{1}{n}\mathbb{E}_{\infty} \left[ \sum_{k=1}^{n-1} \theta_k  \right] \leq \epsilon.
\end{align}
To cope with this constraint, we adopt censoring strategy, whereas \cite{banerjee2012data} resorts to detection procedure and they propose a so called ``DE-CuSum" algorithm to detect the change event.

Let us illustrate with a simple example why the quickest change detection problem with energy constraint is interesting also from a practical viewpoint. Consider the structural health monitoring of a bridge \cite{rice2010flexible}. Wireless battery-powered sensor nodes are deployed to monitor the bridge. Based on the collected data, a person in charge aims to detect any abnormal condition (e.g., a small crack in the bridge) as soon as possible to take suitable actions. The false alarm rate should be as small as possible, as unnecessary actions are costly. Note that the reciprocal of the ARL is connected to the false alarm rate. The energy constraint \eqref{Eqn:energyconstraint} is natural as the battery-powered sensors might be difficult and costly to recharge. For the overall system design, the energy constraint can be viewed as a design parameter, which should be considered together with the likelihood of an abnormal event and the false alarm rate.

Define the filtration $\{\mathcal{F}^{\mu}_k\}_{k\in\mathbb{N}}$ induced by the observations $\{X_k\}_{k\in\mathbb{N}_{+}}$ and the decision rule $\mu$ as
\[\mathcal{F}^{\mu}_k \triangleq \sigma(\theta_1X_1,\ldots,\theta_kX_k,\theta_1,\ldots,\theta_k), \forall k\in\mathbb{N}_{+}, \]
with the understanding that $\mathcal{F}^{\mu}_0$ is the trivial $\sigma$-algebra.
Similar to the classical setup, we define the following quantities:
\begin{align}
\mathcal{D}^{\mu}_L(T) \triangleq & \sup_{0\leq\nu<\infty} \big\{ \esssup \mathbb{E}^{\mu}_{\nu}[(T-\nu)^+|\mathcal{F}^{\mu}_{\nu}] \big\},  \label{Eqn:DetectionDelayLorden}  \\
\mathcal{D}^{\mu}_P(T) \triangleq & \sup_{0\leq\nu<\infty} \mathbb{E}^{\mu}_{\nu}[(T-\nu)^+|T>\nu],\\
\mathcal{C}^{\mu}_{\gamma}\triangleq & \, \{T:\mathbb{E}^{\mu}_{\infty}[T]\geq \gamma\},
\end{align}
where $\mathbb{E}^{\mu}_{\nu}$ is the expectation when the change event happens at time $\nu$ and the decision policy $\mu$ is adopted; $\mathbb{E}^{\mu}_{\infty}$, $\mathbb{P}^{\mu}_{\nu}$ and $\mathbb{P}^{\mu}_{\infty}$ are defined similarly. Denoted by $\mathcal{U}_{\epsilon}$ the class of all the admissible decision policies:
\[\mathcal{U}_{\epsilon} \triangleq \, \{ \mu: \mathbb{E}_{\infty}[\mu(X_1)] \leq \epsilon \}. \]
 Then the two problems corresponding to the classical Lorden's and Pollak's problem we are interested in are formulated as follows:
\begin{problem} \label{Problem:Lorden}
\begin{align*}
\text{find} &  \quad \mu^{*}, T^{*} \\
\text{s.t.} &  \quad \mathcal{D}^{\mu^{*}}_L(T^*)=\inf_{T\in\mathcal{C}_{\gamma},\mu\in\mathcal{U}_{\epsilon}}\mathcal{D}^{\mu}_L(T), \:\text{for every}\: \gamma>1.
\end{align*}
\end{problem}
and
\begin{problem}  \label{Problem:Pollak}
\begin{align*}
\text{find} &  \quad \mu^{*}, T^{*} \\
\text{s.t.} &  \quad \mathcal{D}^{\mu^{*}}_P(T^*)=\inf_{T\in\mathcal{C}_{\gamma},\mu\in\mathcal{U}_{\epsilon}}\mathcal{D}^{\mu}_P(T), \:\text{for every}\: \gamma>1.
\end{align*}
\end{problem}

\section{Main Results} \label{Section:Main results}
In this section, we provide solutions to the above two problems. We first prove that the censoring strategy that has the maximal post-censoring K-L divergence coupled with the CuSum algorithm and SRP procedure is asymptotically optimal for Problem \ref{Problem:Lorden} and Problem \ref{Problem:Pollak}, respectively. In general, to compute such asymptotically optimal censoring strategy can only be done numerically. If we search for the asymptotically optimal censoring strategy in the whole admissible class $\mathcal{U}_{\epsilon}$, the computation load will be very huge, especially when the dimension of the observation is high. To alleviate this computation burden, we give two necessary conditions on the asymptotically optimal censoring strategy. One is that the asymptotically optimal censoring strategy should use up the available energy, i.e., equation \eqref{Eqn:energyconstraint} becomes equality for the asymptotically optimal censoring strategy. The other is that the asymptotically optimal censoring strategy has a special structure, that is, the likelihood ratio of the \emph{no send} region is a single interval.

\subsection{Asymptotically Optimal Censoring Strategy for Lorden's Problem}
Define a variation of likelihood ratio function as
\begin{align} \label{Eqn:likelihoodratio}
L^{\mu}(X_k,\theta_k) \triangleq \left\{
\begin{array}{ll}
(\mathrm{d}\mathbb{P}_1/\mathrm{d}\mathbb{P}_{\infty})(X_k), & \text{if }  \theta_k=1,\\
\frac{\mathbb{P}^{\mu}_1\{ \theta_k=0\} }{\mathbb{P}^{\mu}_{\infty}\{ \theta_k=0\}}, & \text{if } \theta_k=0,
\end{array} \right.
\end{align}
where $\mathrm{d}\mathbb{P}_1/\mathrm{d}\mathbb{P}_{\infty}$ is the Radon-Nikodym derivative. Note that since it is assumed that $\mathbb{P}_{\infty}$ and $\mathbb{P}_1$ are mutually locally absolutely continuous, such derivative always exists.
We introduce a statistic based on $L^{\mu}(X_k,\theta_k)$, which can be calculated recursively by
\begin{align}
S_k=&\max_{1\leq q \leq k} \, \prod_{i=q}^{k}L^{\mu}(X_i,\theta_i)\\
=&\max(S_{k-1},1)L^{\mu}(X_k,\theta_k), \forall k\in\mathbb{N}_{+},
\end{align}
with $S_0=0$.
The stopping time based on $S_k$ is given by
\begin{align}
T_{\gamma}^A=\inf\{k \geq 0 : S_k \geq A \},  \label{Eqn:StoppingtimeCUSUM}
\end{align}
where $A$ is a constant threshold such that
\begin{align}
\mathbb{E}^{\mu}_{\infty}[T_{\gamma}^A] = \gamma.
\end{align}
When defining stopping times, we adopt the convention that $\inf\{\emptyset\}=\infty$, i.e., $T_{\gamma}^A=\infty$ if the statistics $S_k$ never crosses $A$.

\begin{lemma} \label{Lemma:optimalstopping}
Given any censoring strategy $\mu \in \mathcal{U}_{\epsilon}$, the stopping time $T_{\gamma}^A$ defined in \eqref{Eqn:StoppingtimeCUSUM} is strictly optimal for Lorden's problem, i.e., for any $\mu$,
\begin{align}
\mathcal{D}^{\mu}_L(T_{\gamma}^A)=\inf_{T\in\mathcal{C}_{\gamma}}\mathcal{D}^{\mu}_L(T), \:\text{for every}\: \gamma>1.
\end{align}
\end{lemma}
\begin{proof}
We introduce the random variable
\begin{align}
Z_k = \left\{
\begin{array}{ll}
X_k, & \text{if }  \theta_k=1,\\
\wp, & \text{if } \theta_k=0,
\end{array} \right.
\end{align}
where $\wp$ is a particular symbol indicating the event when the center does not receive the data from the remote sensor. ``Receiving nothing" can be regarded as a special observation, since the censoring strategy is observations dependent. Instead of just discarding such special observations (assign likelihood ratio $1$ in the recursion update for Page's CuSum statistics), we assign the corresponding likelihood ratio value according to the policy used. Note that $\wp$ should be chosen such that $\mathbb{P}_1(\wp) = \mathbb{P}_{\infty}(\wp) = 0.$

Since the censoring strategy is stationary, it is easily seen that if the change event happens at $\nu$ and the particular policy $\mu$ is adopted, $Z_1^{\nu-1}$ are i.i.d. with measure $\mathbb{P}^{\mu}_{\infty}$ and $Z_{\nu}^{\infty}$ are i.i.d. with density function $\mathbb{P}^{\mu}_1$. Denote by $\{\overline{\mathcal{F}}_k\}_{k\in\mathbb{N}}$ the filtration induced by the observations $\{Z_k\}_{k\in\mathbb{N}_{+}}$, defined similar to $\{\mathcal{F}_k\}_{k\in\mathbb{N}}$. The detection delay associated with Lorden's problem defined in \eqref{Eqn:DetectionDelayLorden} can be rewritten as
\[  \mathcal{D}^{\mu}_L(T) \triangleq  \sup_{0\leq\nu<\infty} \big\{ \esssup \mathbb{E}^{\mu}_{\nu}[(T-\nu)^+|\overline{\mathcal{F}}_{\nu}] \big\}.  \]
Then given any admissible censoring strategy $\mu \in \mathcal{U}_{\epsilon}$, Problem \ref{Problem:Lorden} can be written as
\begin{align*}
\text{find} & \quad  T^{*} \\
s.t. &  \quad \mathcal{D}^{\mu}_L(T^*)=\inf_{T\in\mathcal{C}_{\gamma}}\mathcal{D}^{\mu}_L(T), \:\text{for every}\: \gamma>1.
\end{align*}
Note that $T$ is also a stopping time with respect to $\overline{\mathcal{F}}_{k}$, so the above problem is just the classical Lorden's formulation. The CuSum procedure has been proved to be strictly optimal for Lorden's formulation \cite{moustakides1986optimal, ritov1990decision}. It is easily verified that $L^{\mu}(X_k,\theta_k)$ defined in \eqref{Eqn:likelihoodratio} is in effect the likelihood ratio function of $Z_k$. Hence the stopping time defined in \eqref{Eqn:StoppingtimeCUSUM} is indeed the Page's CuSum procedure based on observations $Z_k$. The strict optimality of $T_{\gamma}^A$ thus follows, which concludes the proof.
\end{proof}

In the following, we study the optimal censoring strategy. To avoid degenerate problem, we assume the finiteness of K-L divergence of local observations in the sequel. Specifically, it is assumed that
\begin{align} \label{Eqn:finiteAssumption}
0<\mathbb{D}(\mathbb{P}_1||\mathbb{P}_{\infty}) \triangleq \mathbb{E}_{k}[\ell(X_k)] < \infty,
\end{align}
where $\ell(X_k)=\ln \frac{\mathrm{d}\mathbb{P}_1}{\mathrm{d}\mathbb{P}_{\infty}}(X_k)$ is the log-likelihood ratio function.

\begin{theorem} \label{Theorem:AsymOptCuSum}
Introduce the censoring strategy
\begin{align}
\mu^{*} \triangleq \argmax_{\mu\in\mathcal{U}_{\epsilon}} \mathbb{E}^{\mu}_{k}[ \ln L^{\mu}(X_k,\theta_k) ]. \label{Eqn:optimalCensoringDef}
\end{align}
Then for any $\epsilon \in (0,1]$, the pair of censoring strategy and stopping time $(\mu^{*},T_{\gamma}^A)$ with $T_{\gamma}^A$ defined in \eqref{Eqn:StoppingtimeCUSUM} is third-order asymptotically optimal for Problem \ref{Problem:Lorden}, i.e.,
\begin{align}  \label{Eqn:LordenAsymOpt}
\mathcal{D}^{\mu^{*}}_L(T_{\gamma}^A)=\inf_{T\in\mathcal{C}_{\gamma},\mu\in\mathcal{U}_{\epsilon}}\mathcal{D}^{\mu}_L(T)+ \smallO 1, \: \text{as}\, \gamma \to \infty.
\end{align}
\end{theorem}
\begin{proof}
The result of Lemma \ref{Lemma:optimalstopping} tells that to find the optimal (regardless of strictly optimal or asymptotically optimal) censoring strategy, instead of searching from all the possible pairs of censoring strategy and stopping time $(\mu,T)$, we can just compare the performance of $(\mu,T_{\gamma}^{A_{\mu}})$. Here we use the notation $A_{\mu}$ instead of $A$ to highlight that the threshold parameter for CuSum depends on the specific censoring strategy being used.

From the invariance properties of K-L divergence \cite[p. 19]{kullback1978information}, no censoring strategy can increase discrimination information, i.e., the post-censoring K-L divergence cannot be larger than the K-L divergence of the local observation:
\begin{align*}
\mathbb{D}(\mathbb{P}^{\mu}_1||\mathbb{P}^{\mu}_{\infty}) \triangleq &\: \mathbb{E}^{\mu}_{k}[ \ln L^{\mu}(X_k,\theta_k) ] \\
\leq & \mathbb{D}(\mathbb{P}_1||\mathbb{P}_{\infty}) < \:\infty.
\end{align*}
From Theorem 3 of \cite{lorden1971procedures}, we can get
\begin{align}
\mathcal{D}^{\mu}_L(T_{\gamma}^{A_{\mu}})=\frac{\ln \gamma}{\mathbb{D}(\mathbb{P}^{\mu}_1||\mathbb{P}^{\mu}_{\infty}) }(1+\smallO1), \text{as}\: \gamma \to \infty.
\end{align}
Let $\mu_1,\mu_2$ be two arbitrary censoring strategies such that
\begin{align}
\mathbb{D}(\mathbb{P}^{\mu_1}_1||\mathbb{P}^{\mu_1}_{\infty}) \geq \mathbb{D}(\mathbb{P}^{\mu_2}||\mathbb{P}^{\mu_2}_{\infty}).
\end{align}
Then as $\gamma \to \infty$,
\begin{align} \label{Eqn:dividelessone}
\frac{\mathcal{D}^{\mu_1}_L(T_{\gamma}^{A_{\mu_1}})}{\mathcal{D}^{\mu_2}_L(T_{\gamma}^{A_{\mu_2}})}=
\frac{\mathbb{D}(\mathbb{P}^{\mu_2}||\mathbb{P}^{\mu_2}_{\infty})}{\mathbb{D}(\mathbb{P}^{\mu_1}_1||\mathbb{P}^{\mu_1}_{\infty})}
\leq  1,
\end{align}
where we define $\frac{0}{0}=1$. If $\mathbb{D}(\mathbb{P}^{\mu}_1||\mathbb{P}^{\mu}_{\infty})=0$, the detection problem is degenerate and the detection delay $\mathcal{D}^{\mu}_L(T_{\gamma}^{A_{\mu}})=\infty$. Hence we treat the class that has zero post-censor discrimination information equally.

The equation \eqref{Eqn:LordenAsymOpt} follows directly from \eqref{Eqn:dividelessone} and the proof thus is complete.
\end{proof}

\subsection{Asymptotically Optimal Censoring Strategy for Pollak's Problem}
Even for the classical Pollak's problem, the strict optimal detection procedure is still an open issue. Thus instead of looking for the optimal (or asymptotically optimal) pair of censoring strategy and detection procedure, we aim to find optimal (or asymptotically optimal) censoring strategy given a specific detection procedure. The SR procedure and its variants: SRP and SR-$r$ have been proved to be asymptotically optimal for the classical Pollak's problem \cite{pollak1985optimal, tartakovsky2012third}. In this paper, we study the asymptotically optimal censoring strategy when the SRP procedure is being used. Note that the other two cases (when the detection procedure is SR or SR-$r$) can be studied similarly.

The SRP procedure for our problem is given by the stopping time
\begin{align}
T_{\gamma}^A=\inf\{k \geq 0 : R_k \geq A \},  \label{Eqn:StoppingtimeSRP}
\end{align}
where $A$ is a constant threshold such that
\begin{align}
\mathbb{E}^{\mu}_{\infty}[T_{\gamma}^A] = \gamma,
\end{align}
and
\begin{align}
R_k=(1+R_{k-1})L^{\mu}(X_k,\theta_k), \: k \geq 1,
\end{align}
with a random initial point $R_0\sim Q_A(x)$, where the quasi-stationary cdf $Q_A(x)$ is defined by
\begin{align} \label{Eqn:QuasiStaDistri}
Q_A(x) \triangleq  \lim_{k\to \infty} \mathbb{P}^{\mu}_{\infty}(R_k \leq x | T_{\gamma}^A > k).
\end{align}
From the classical quickest change detection theory, we know that given any censoring strategy $\mu$, the corresponding SRP is third-order asymptotically optimal, i.e.,
\begin{align}
\mathcal{D}^{\mu}_P(T_{\gamma}^A)=\inf_{T\in\mathcal{C}_{\gamma}}\mathcal{D}^{\mu}_P(T)+\smallO 1, \:\text{as}\: \gamma \to \infty.
\end{align}
The problem we are interested in is that given the stopping time $T_{\gamma}^A$ defined in \eqref{Eqn:StoppingtimeSRP}, what is the asymptotically optimal censoring strategy? The result is stated in the following theorem. To study the properties of SRP procedure, in this subsection we assume that $\mathbb{E}_{k}[{\ell(X_k)}^2] < \infty$, and there is no point mass for $\ell(X_k)$ under either $\mathbb{P}_{1}$ or $\mathbb{P}_{\infty}$, i.e.,
\begin{align}
\mathbb{P}_{1}(\ell(X_k)=t)=0, \, & \forall t\geq 0,\\
\mathbb{P}_{\infty}(\ell(X_k)=t)=0, \, & \forall t\geq 0.
\end{align}

\begin{lemma} \label{Lemma:finitesecondlikelistrategy}
$\mathbb{E}_{k}[{\ell(X_k)}^2] < \infty$ implies \[\mathbb{E}^{\mu}_{k}[ \ln^2 L^{\mu}(X_k,\theta_k)]< \infty, \forall \mu \in \mathcal{U}_{\epsilon}. \]
\end{lemma}
\begin{proof}
If $\mathbb{E}_{\infty}[\mu(X_k)] =1$, then
\[\mathbb{E}^{\mu}_{k}[ \ln^2 L^{\mu}(X_k,\theta_k)]=\mathbb{E}_{k}[{\ell(X_k)}^2]< \infty.\]
If $\mathbb{E}_{\infty}[\mu(X_k)] <1$, let $\Omega_c$ be the censoring region associated with $\mu$.
Since $\mathbb{P}_{\infty}$ and $\mathbb{P}_1$ are mutually locally absolutely continuous measures, then for any $\Omega_c$,
\begin{align*}
0<\frac{\int_{\Omega_{c}}\mathrm{d}\mathbb{P}_1}{\int_{\Omega_{c}}\mathrm{d}\mathbb{P}_{\infty}}<\infty.
\end{align*}
We then obtain
\begin{align*}
&\mathbb{E}^{\mu}_{k}[ \ln^2 L^{\mu}(X_k,\theta_k)] \\
=& \int_{\Omega\backslash \Omega_{c}}\ln^2\frac{\mathrm{d}\mathbb{P}_1}{\mathrm{d}\mathbb{P}_{\infty}} \mathrm{d}\mathbb{P}_1 +
\int_{\Omega_{c}}\ln^2\frac{\int_{\Omega_{c}}\mathrm{d}\mathbb{P}_1}{\int_{\Omega_{c}}\mathrm{d}\mathbb{P}_{\infty}} \mathrm{d}\mathbb{P}_1\\
\leq & \mathbb{E}_{k}[{\ell(X_k)}^2] +\ln^2\frac{\int_{\Omega_{c}}\mathrm{d}\mathbb{P}_1}{\int_{\Omega_{c}}\mathrm{d}\mathbb{P}_{\infty}} \\
< & \infty,
\end{align*}
and the proof is complete.
\end{proof}

\begin{theorem} \label{Theorem:AsymOptSRP}
Let $\mu^{*}$ be the censoring strategy defined in \eqref{Eqn:optimalCensoringDef}. Then given the stopping time $T_{\gamma}^A$ defined in \eqref{Eqn:StoppingtimeSRP}, for any $\epsilon\in(0,1]$, $\mu^{*}$ is third-order asymptotically optimal. Specifically,
\begin{align} \label{Eqn:SRPgreaterKLopt}
\mathcal{D}^{\mu^{*}}_P(T_{\gamma}^A)=\inf_{\mu\in\mathcal{U}_{\epsilon}}\mathcal{D}^{\mu}_P(T_{\gamma}^A)+\smallO1, \:\text{as}\: \gamma\to\infty.
\end{align}
\end{theorem}
\begin{proof}
Given any censoring strategy $\mu$, since
\[\mathbb{E}_{k}[{\ell(X_k)}^2] < \infty,\]
by Lemma \ref{Lemma:finitesecondlikelistrategy}, one obtains
\[\mathbb{E}^{\mu}_{k}[ \ln^2 L^{\mu}(X_k,\theta_k)]< \infty.\]
Also as there is no point mass for $\ell(X_k)$, $\ln L^{\mu}(X_k,\theta_k)$ is non-arithmetic \footnote{A random variable $Y \in \mathbb{R}$ is said to be arithmetic if there exists constant $d>0$ such that \[\text{Pr}\{Y\in\{\ldots,-2d,-d,0,d,2d,\ldots\}\}=1.\] Otherwise it is called non-arithmetic. } for any censoring strategy. Then from \cite{tartakovsky2012third}, one obtains that
\begin{align}
\mathcal{D}^{\mu}_P(T_{\gamma}^A)=&\frac{1}{\mathbb{D}(\mathbb{P}^{\mu}_1||\mathbb{P}^{\mu}_{\infty})}(\ln A + \aleph -C_{\infty}) + \smallO 1\: \text{as}\: A \to \infty,  \label{Eqn:SRPasymperDelay} \\
\mathbb{E}_{\infty}^{\mu}(T_{\gamma}^A)=&\frac{A}{\zeta}(1+\smallO 1) \: \text{as}\: A \to \infty, \label{Eqn:SRPasymperARL}
\end{align}
where $\aleph$, $C_{\infty}$ and $\zeta$ are given as follows. Let $S_n=\ln L^{\mu}(X_1,\theta_1)+ \cdots + \ln L^{\mu}(X_n,\theta_n)$ and define the one-sided stopping time by $\tau_a=\arg\min\{n\geq 1: S_n \geq a\},$ for $a\geq 0$. Let $\kappa_a=S_{\tau_a}-a$ be the excess over the threshold $a$ at the stopping time, then $\aleph$ and $\zeta$ in \eqref{Eqn:SRPasymperDelay} and \eqref{Eqn:SRPasymperARL} is defined by
\begin{align*}
\aleph=\lim_{a\to\infty}\mathbb{E}_{1}^{\mu}[\kappa_a],
\qquad \qquad \zeta=\lim_{a\to\infty}\mathbb{E}_{1}^{\mu}[e^{-\kappa_a}].
\end{align*}
The variable $C_{\infty}$ is given by
\begin{align*}
C_{\infty}=\mathbb{E}_{\infty}^{\mu}[\ln(1+R_{\infty}+V_{\infty})],
\end{align*}
where $V_{\infty}=\sum_{i=1}^{\infty}e^{-S_i}$, and $R_{\infty}$ is a random variable that has the $\mathbb{P}_{\infty}^{\mu}$-limiting distribution of $R_n$ as $n\to\infty$, i.e.,
\[\lim_{n\to\infty}\mathbb{P}_{\infty}^{\mu}(R_n\leq x)=\mathbb{P}_{\infty}^{\mu}(R_{\infty}\leq x).\]
Obviously $0<\zeta<1$, and from the renewal theory (e.g., \cite{siegmund1985sequential}), one can obtain that $0<\aleph<\infty$ and $0<C_{\infty}<\infty$ whatever the censoring strategy is. Let $A=\gamma \zeta$, and since $\aleph$, $\zeta$ and $C_{\infty}$ only depend on the censoring strategy and observation model, which is independent of $\gamma$, we can rewrite \eqref{Eqn:SRPasymperDelay} and \eqref{Eqn:SRPasymperARL} as
\begin{align}
\mathcal{D}^{\mu}_P(T_{\gamma}^A)=&\frac{\ln \gamma}{\mathbb{D}(\mathbb{P}^{\mu}_1||\mathbb{P}^{\mu}_{\infty})}(1 + \smallO 1) \: \text{as}\: \gamma \to \infty, \\
\mathbb{E}_{\infty}^{\mu}(T_{\gamma}^A)=&\gamma(1+\smallO 1) \: \text{as}\: \gamma \to \infty,
\end{align}
Let $\mu_1,\mu_2$ be two arbitrary censoring strategies such that
\begin{align}
\mathbb{D}(\mathbb{P}^{\mu_1}_1||\mathbb{P}^{\mu_1}_{\infty}) \geq \mathbb{D}(\mathbb{P}^{\mu_2}||\mathbb{P}^{\mu_2}_{\infty}).
\end{align}
Then as $\gamma \to \infty$,
\begin{align} \label{Eqn:dividelessone1}
\frac{\mathcal{D}^{\mu_1}_P(T_{\gamma}^{A_{\mu_1}})}{\mathcal{D}^{\mu_2}_P(T_{\gamma}^{A_{\mu_2}})}=&
\frac{\mathbb{D}(\mathbb{P}^{\mu_2}||\mathbb{P}^{\mu_2}_{\infty})}{\mathbb{D}(\mathbb{P}^{\mu_1}_1||\mathbb{P}^{\mu_1}_{\infty})}
\leq  1, \\
\frac{\mathbb{E}_{\infty}^{\mu_1}(T_{\gamma}^A)}{\mathbb{E}_{\infty}^{\mu_2}(T_{\gamma}^A)}=&1.
\end{align}
Equation \eqref{Eqn:SRPgreaterKLopt} follows, so the proof is complete.
\end{proof}

\subsection{Maximize K-L Divergence} \label{Section:maximizeKLDivergence}
The above theorems show that the asymptotically optimal censoring strategy for both Lorden's and Pollak's problem is the one that has the maximal post-censoring K-L divergence.
In general, to find the optimal censoring strategy that maximizes the post-censoring K-L divergence has to be done numerically. In this subsection, we give two properties that the optimal strategy possesses. The usefulness of our results is that these two properties can be utilized to significantly reduce the computation load. Before proceeding, we introduce the following assumption.
\begin{assumption} \label{Assumption:nomasspoint}
There is no point mass for observation $X_k$ under either $\mathbb{P}_{\infty}$ or $\mathbb{P}_{1}$, i.e., both $\mathbb{P}_{\infty}$ and $\mathbb{P}_{1}$ are continuous over $\mathcal{F}$.
\end{assumption}
\begin{remark}
Under this assumption, our argument and presentation will be simplified. For scenarios where the observations have point mass, randomized censoring strategy can be used instead. For a randomized censoring strategy, whether a observation taken at the sensor is sent to the fusion center or not depends on not only the observation itself but also on another random variable (which needs to be carefully defined to meet certain constraints). Randomization over these possible point masses can ``split" them arbitrarily, so the problem will be reduced to one with no point mass. The following theorems thus can be extended to the point mass case easily.
\end{remark}
\begin{lemma} \label{Lemma:greaterratio}
Suppose $\mathbf{P}$ and $\mathbf{Q}$ are two probability measures on a measurable space $(\Omega,\mathcal{F})$.
Given any $A_1 \in \mathcal{F}$, define the following two quantities:
\begin{align}
|A_1|_{\mathbf{Q}\mathbf{P}} \triangleq & \frac{\int_{A_1} \mathrm{d}\mathbf{Q}}{\int_{A_1} \mathrm{d}\mathbf{P}},\\
|A_1|_{\mathbf{P}} \triangleq & \int_{A_1} \mathrm{d}\mathbf{P}.
\end{align}
Then $\forall~ b \in[0,|A_1|_{\mathbf{P}}]$, we can always find $A_2 \subseteq A_1$ such that
\begin{align}
|A_2|_{\mathbf{Q}\mathbf{P}} \geq & |A_1|_{\mathbf{Q}\mathbf{P}},\\
|A_2|_{\mathbf{P}}= & b,
\end{align}
if \begin{enumerate}
     \item $\mathbf{Q}$ is absolutely continuous with respect to $\mathbf{P}$.
     \item $\mathbf{P}$ is continuous.
   \end{enumerate}
\end{lemma}
\begin{proof}
Since $\mathbf{Q}$ is absolutely continuous w.r.t. $\mathbf{P}$, we can obtain the Radon-Nikodym derivative as
\[y(\omega)=\frac{\mathrm{d}\mathbf{Q}}{\mathrm{d}\mathbf{P}}(\omega).\]
Then the likelihood ratio for $A_1$ can be rewritten by
\begin{align*}
|A_1|_{\mathbf{Q}\mathbf{P}} = & \frac{\int_{A_1} \mathrm{d}\mathbf{Q}}{\int_{A_1} \mathrm{d}\mathbf{P}}=  \frac{\int_{A_1} y \mathrm{d}\mathbf{P}}{\int_{A_1} \mathrm{d}\mathbf{P}},
\end{align*}
and it follows that
\begin{align}
\max_{\omega \in A_1}y(\omega) \geq & |A_1|_{\mathbf{Q}\mathbf{P}}, \\
\min_{\omega \in A_1}y(\omega) \leq & |A_1|_{\mathbf{Q}\mathbf{P}}.
\end{align}
If
 $\max_{\omega \in A_1}y(\omega) =  |A_1|_{\mathbf{Q}\mathbf{P}}$
we can obtain
\[y(\omega)=|A_1|_{\mathbf{Q}\mathbf{P}}, \: \forall \omega \in A_1.  \]
The problem becomes trivial.

In the following, we consider the case where
$\max_{\omega \in A_1}y(\omega) >  |A_1|_{\mathbf{Q}\mathbf{P}}.$
Define
\[\overline{a} \triangleq \max_{\omega \in A_1}y(\omega). \]
Let
\begin{align} \label{Eqn:A2}
A_2=\{\omega \in A_1: y(\omega)\in [c,\overline{a}]\}.
\end{align}
Since $\mathbf{P}$ is continuous, we can always find an appropriate $c$ such that
\[|A_2|_{\mathbf{P}}=  b.\]
We now prove that $A_2$ defined in \eqref{Eqn:A2} satisfies $|A_2|_{\mathbf{Q}\mathbf{P}} \geq  |A_1|_{\mathbf{Q}\mathbf{P}}$. If $ c\geq |A_1|_{\mathbf{Q}\mathbf{P}}$,
\[|A_2|_{\mathbf{Q}\mathbf{P}}=\frac{\int_{A_2} y \mathrm{d}\mathbf{P}}{\int_{A_2} \mathrm{d}\mathbf{P}} \geq c \frac{\int_{A_1} \mathrm{d}\mathbf{P}}{\int_{A_2} \mathrm{d}\mathbf{P}} \geq |A_1|_{\mathbf{Q}\mathbf{P}}.\]
If $ c < |A_1|_{\mathbf{Q}\mathbf{P}}$, we obtain
\[|A_1\backslash A_2|_{\mathbf{Q}\mathbf{P}}=\frac{\int_{A_1\backslash A_2} y \mathrm{d}\mathbf{P}}{\int_{A_1\backslash A_2} \mathrm{d}\mathbf{P}} \leq c \frac{\int_{A_1} \mathrm{d}\mathbf{P}}{\int_{A_1} \mathrm{d}\mathbf{P}} < |A_1|_{\mathbf{Q}\mathbf{P}}.\]
Suppose $|A_2|_{\mathbf{Q}\mathbf{P}} < |A_1|_{\mathbf{Q}\mathbf{P}}$. It then follows that
\begin{align*}
|A_1|_{\mathbf{Q}\mathbf{P}}=&\frac{ \int_{A_2} \mathrm{d}\mathbf{Q} +\int_{A_1\backslash A_2} \mathrm{d}\mathbf{Q}  }{  \int_{A_1} \mathrm{d}\mathbf{P} + \int_{A_1\backslash A_2} \mathrm{d}\mathbf{P}   }\\
< &|A_1|_{\mathbf{Q}\mathbf{P}} \frac{ \int_{A_1} \mathrm{d}\mathbf{P} + \int_{A_1\backslash A_2} \mathrm{d}\mathbf{P}  }{  \int_{A_1} \mathrm{d}\mathbf{P} + \int_{A_1\backslash A_2} \mathrm{d}\mathbf{P}   } \\
=&|A_1|_{\mathbf{Q}\mathbf{P}}.
\end{align*}
Thus $|A_2|_{\mathbf{Q}\mathbf{P}} \geq |A_1|_{\mathbf{Q}\mathbf{P}}$. The proof thus is complete.
\end{proof}

\begin{theorem} \label{Theorem:EqualGreaterKL}
Given any censoring strategy $\underline{\mu} \in \mathcal{U}_{\epsilon}$, we can always find another strategy $\overline{\mu} \in \mathcal{U}_{\epsilon}$ that satisfies
\begin{align}
\mathbb{E}_{\infty}[\overline{\mu}(X_k)]~ =~ &\epsilon, \\
\mathbb{E}^{\overline{\mu}}_{k}[ \ln L^{\overline{\mu}}(X_k,\theta_k) ]~ \geq ~& \mathbb{E}^{\underline{\mu}}_{k}[ \ln L^{\underline{\mu}}(X_k,\theta_k) ]. \label{Eqn:BiggerK-L}
\end{align}
\end{theorem}
\begin{proof}
A censoring strategy can be fully characterised by its \emph{send} region or \emph{no send} region. Suppose strategy $\underline{\mu}$ is given by
\begin{align}
\theta_k=\left\{
\begin{array}{ll}
0, & \text{if }  X_k \in \Omega_{\underline{\mu}},\\
1, & \text{otherwise},
\end{array} \right.
\end{align}
where $\Omega_{\underline{\mu}} \in \Omega$ is the \emph{no send} region associated with the strategy $\underline{\mu}$.
It follows from $\underline{\mu} \in \mathcal{U}_{\epsilon}$ that
\begin{align}
\int_{\Omega_{\underline{\mu}}} \mathrm{d}\mathbb{P}_{\infty} \geq 1-\epsilon.
\end{align}
Then under Assumption \ref{Assumption:nomasspoint} and by Lemma \ref{Lemma:greaterratio}, we can always find $\Omega_{\overline{\mu}} \subseteq \Omega_{\underline{\mu}}$ such that
\begin{align}
\int_{ \Omega_{\overline{\mu}}} \mathrm{d}\mathbb{P}_{\infty} =& 1-\epsilon, \\
\frac{ \int_{\Omega_{\overline{\mu}}} \mathrm{d}\mathbb{P}_{1} }{ \int_{\Omega_{\overline{\mu}}} \mathrm{d}\mathbb{P}_{\infty} } \geq & \frac{ \int_{\Omega_{\underline{\mu}}} \mathrm{d}\mathbb{P}_{1} }{ \int_{\Omega_{\underline{\mu}}} \mathrm{d}\mathbb{P}_{\infty} } \label{Eqn:greaterratioTheo}
\end{align}
Define strategy $\overline{\mu}$ by
\begin{align}
\theta_k=\left\{
\begin{array}{ll}
0, & \text{if }  X_k \in \Omega_{\overline{\mu}},\\
1, & \text{otherwise}.
\end{array} \right.
\end{align}
We obtain
\begin{align*}
&\mathbb{E}^{\underline{\mu}}_{k}[ \ln L^{\underline{\mu}}(X_k,\theta_k) ] - \mathbb{E}^{\overline{\mu}}_{k}[ \ln L^{\overline{\mu}}(X_k,\theta_k) ] \\
= & \int_{ \Omega_{\underline{\mu}} } \ln \frac{ \int_{\Omega_{\underline{\mu}}}\mathrm{d}\mathbb{P}_{1}  }{ \int_{\Omega_{\underline{\mu}}} \mathrm{d}\mathbb{P}_{\infty} } \mathrm{d}\mathbb{P}_{1}     +
\int_{ \Omega\backslash \Omega_{\underline{\mu}} } \ln \frac{ \mathrm{d}\mathbb{P}_{1} }{ \mathrm{d}\mathbb{P}_{\infty}} \mathrm{d}\mathbb{P}_{1} \\
&-\, \int_{ \Omega_{\overline{\mu}} } \ln \frac{ \int_{\Omega_{\overline{\mu}}}\mathrm{d}\mathbb{P}_{1}  }{ \int_{\Omega_{\overline{\mu}}} \mathrm{d}\mathbb{P}_{\infty} } \mathrm{d}\mathbb{P}_{1}     -
\int_{ \Omega\backslash \Omega_{\overline{\mu}} } \ln \frac{ \mathrm{d}\mathbb{P}_{1} }{ \mathrm{d}\mathbb{P}_{\infty}} \mathrm{d}\mathbb{P}_{1} \\
=& \int_{ \Omega_{\overline{\mu}} } \left[ \ln \frac{ \int_{\Omega_{\underline{\mu}}}\mathrm{d}\mathbb{P}_{1}  }{ \int_{\Omega_{\underline{\mu}}} \mathrm{d}\mathbb{P}_{\infty} } - \ln \frac{ \int_{\Omega_{\overline{\mu}}}\mathrm{d}\mathbb{P}_{1}  }{ \int_{\Omega_{\overline{\mu}}} \mathrm{d}\mathbb{P}_{\infty} }\right]\mathrm{d}\mathbb{P}_{1} \\
& +\, \int_{ \Omega_{\underline{\mu}}\backslash \Omega_{\overline{\mu}} } \ln \frac{ \int_{\Omega_{\underline{\mu}}}\mathrm{d}\mathbb{P}_{1}  }{ \int_{\Omega_{\underline{\mu}}} \mathrm{d}\mathbb{P}_{\infty} } \mathrm{d}\mathbb{P}_{1}  \\
& -\, \int_{ \Omega_{\underline{\mu}}\backslash \Omega_{\overline{\mu}} } \ln \frac{ \mathrm{d}\mathbb{P}_{1} }{ \mathrm{d}\mathbb{P}_{\infty}} \mathrm{d}\mathbb{P}_{1} \\
\leq& 0,
\end{align*}
where the inequality follows from \eqref{Eqn:greaterratioTheo} and the invariance properties of K-L divergence. The proof thus is complete.
\end{proof}
\begin{remark}
The intuition behind the above theorem is that the more energy the sensor uses for communication with the decision maker, the bigger K-L divergence of available observations and hence the better asymptotic detection performance can be obtained.
\end{remark}

In the following, we will show that the asymptotically optimal censoring strategy that maximizes the post-censoring K-L divergence has a very special structure, i.e., the likelihood ratio of the \emph{no send} region is a single interval. As the proof relies on the optimal quantization structure established in \cite{tsitsiklis1993extremal}, we will introduce the concept of randomized likelihood-ratio quantizer (RLRQ) first.
\begin{definition}
Let the threshold vector be $\vec{t}=(t_1,\ldots,t_{D-1}) \in \mathbb{R}_{+}^{D-1}$ with $0 \leq t_1 \leq \cdots \leq t_{D-1} \leq \infty$ and the associated real-valued random vector be $
\vec{r}=(r_1,\ldots,r_{D-1})  \in \mathbb{R}^{D-1}$. The elements $r_1,\ldots,r_{D-1}$ are independent of each other. The intervals associated with $\vec{t}$ are defined by $I_1=(0,t_1), I_2=(t_1,t_2),\ldots,I_{D-1}=(t_{D-2},t_{D-1}), I_{D}=(t_{D-1},\infty)$.

A quantizer $\phi: \Omega \mapsto \{1,\ldots,D\}$ is a monotone RLRQ with threshold vector $\vec{t}$ and random vector $\vec{r}$ if
\begin{align*}
&\text{Pr}\left( L(\omega)\in I_d \,\, \& \,\, \phi(\omega)\neq d \right)=0, \quad \forall d, \\
&\mathbf{1}_{\{L(\omega)=t_d\,\, \& \,\, \phi(\omega)=d\,\, \& \,\, r_d\in\mathbb{R}_d \}}\\
&\,+\,\mathbf{1}_{\{L(\omega)=t_d\,\, \& \,\, \phi(\omega)=d+1\,\, \& \,\, r_d\in\mathbb{R}\backslash\mathbb{R}_d\}}=1, \quad \forall d , \alpha_d,
\end{align*}
where $L(\cdot)$ is the likelihood ratio function and $\mathbb{R}_d \subseteq \mathbb{R}$ is the ``selection" set, which together with the random variable $r_d$ determines the quantization output of those points that have the likelihood ratio on the boundary $t_d$.

A quantizer $\phi$ is defined to be a RLRQ if there exists a permutation map $\pi$ such that $\pi \circ \phi$ is a monotone RLRQ.
\end{definition}
\begin{remark}
With the above definition, the RLRQ reduces to the deterministic one when the likelihood ratio $L(\omega)$ belongs to the interior of the intervals. The quantizer is randomized with the aid of carefully designed random variable $r_d$ and associated selection set $\mathbb{R}_d$ on the boundary $t_d$.
\end{remark}


\begin{theorem} \label{Theorem:likelihoodratio}
The following randomized likelihood-ratio-based censoring strategy can achieve the maximal post-censoring K-L divergence defined in \eqref{Eqn:optimalCensoringDef}: $\forall k$,
\begin{align}  \label{Eqn:optcensor}
\theta_k=\left\{
\begin{array}{ll}
0, & \text{if }  \underline{L}_c<L(X_k)<\overline{L}_{c} ,\\
1, & \text{if }  L(X_k)<\underline{L}_{c}\: \text{or}\: L(X_k)>\overline{L}_{c},
\end{array} \right.
\end{align}
and if $X_k$ is in the boundary, i.e., $L(X_k)=\underline{L}_c$ or $\overline{L}_{c}$, then $\theta_k$ is determined by not only $L(X_k)$ but also the auxiliary independent random variables $\kappa_1, \kappa_2 \in \mathbb{R}$, respectively. Specifically, when $L(X_k)=\underline{L}_c$,
\begin{align}
\theta_k=\left\{
\begin{array}{ll}
0, & \text{if }  \kappa_1 \in \mathbb{R}_c^{1} ,\\
1, & \text{otherwise},
\end{array} \right.
\end{align}
and when
$L(X_k)=\overline{L}_c$,
\begin{align}
\theta_k=\left\{
\begin{array}{ll}
0, & \text{if }  \kappa_2 \in \mathbb{R}_c^{2} ,\\
1, & \text{otherwise}.
\end{array} \right.
\end{align}
The threshold parameters $\underline{L}_c, \overline{L}_c$ and the censoring region for the auxiliary random variable $\mathbb{R}_c^{1}, \mathbb{R}_c^{2}$ should be chosen such that
\begin{align*}
&\int_{L(X_1)\in(\underline{L}_c, \overline{L}_c)} \mathrm{d}\mathbb{P}_{\infty}+ \mathbb{P}_{\infty}(L(X_1)=\underline{L}_c) \text{Pr}\{\kappa_1\in\mathbb{R}_c^{1}\} \\
& +\,\mathbb{P}_{\infty}(L(X_1)=\overline{L}_c) \text{Pr}\{\kappa_2\in\mathbb{R}_c^{2}\} \,=\, 1- \epsilon. \addtag \label{Eqn:EqualCons}
\end{align*}
\end{theorem}
\begin{proof}
The equality constraint in \eqref{Eqn:EqualCons} follows directly from Theorem \ref{Theorem:EqualGreaterKL}. We thus focus on the proof that the optimal censoring strategy has the structure described in \eqref{Eqn:optcensor}.

In the setting of censoring, it is assumed that when the sensor decides to send its observation, the ``real" observations are sent without any quantization. To prove the likelihood-ratio-based structure, we first assume that the sensor sends quantized observations instead of unquantized ones. Specifically, if the observations are deemed as ``uninformative", the sensor sends nothing; the sensor sends one of the $D$ symbols otherwise. Mathematically, the sensor adopts a quantization rule as:
\[\phi: \Omega \mapsto \{0,1,\ldots,D\},\]
where if the observation is mapped to $0$, it means that the sensor sends nothing. Then we study the structure of the optimal quantization rule that maximizes the K-L divergence after quantization, which is given by
\begin{align}
\mathbb{D}^{q} \triangleq \sum_{d=0}^{D} \mathbb{P}_{1}(\phi(\omega)=i)   \ln \frac{\mathbb{P}_{1}(\phi(\omega)=i)}{\mathbb{P}_{\infty}(\phi(\omega)=i)},
\end{align}
subject to
\begin{align}
\mathbb{P}_{\infty}(\phi(\omega)=0)=1-\epsilon.
\end{align}
Because of the convexity of the K-L divergence \cite[P. $32$]{cover2006elements}, under the finiteness assumption in \eqref{Eqn:finiteAssumption}, from Proposition $3.5$ in \cite{tsitsiklis1993extremal}, we know that the above optimal quantization rule has randomized likelihood-ratio-based structure. Note that the censoring strategy can be regarded as the special quantization case where $D=\infty$. Since the randomized likelihood ratio structure holds for every finite $D$, so does the censoring strategy. The proof thus is complete.
\end{proof}

\begin{remark}
We assume there is no point mass for $X_k\in\Omega$ under either $\mathbb{P}_{1}$ or $\mathbb{P}_{\infty}$, but it is likely that there exists point mass for the likelihood ratio $L(X_k)$. We hence consider splitting those points that belongs to the boundary by randomization. We should also note that for any randomized likelihood-ratio-based censoring strategy, there always  exists a deterministic observation-based strategy that has the same post-censoring K-L divergence. If there is no point mass in the boundary, i.e., $\mathbb{P}_{\infty}(L(X_k)=\underline{L}_c)=0 $ and $\mathbb{P}_{\infty}(L(X_k)=\overline{L}_c)=0$, both the likelihood-ratio-based and observation-based strategies become deterministic and the optimal censoring strategy is unique; whereas there are an infinite observation-based strategies otherwise.
\end{remark}
\begin{remark}
Rago et al. \cite{rago1996censoring} first introduced the concept of censoring strategy. The authors proved that with several different detection performance indices, the likelihood ratio of the \emph{no send} region is one single interval. In particular, the authors stated in Theorem 3 of \cite{rago1996censoring} that to maximize the Ali-Silvey distance, the \emph{no send} region should be one single interval, which is very similar to our result.
An important limitation of \cite{rago1996censoring} is the assumption that there is no point mass for the likelihood ratio function. Our approach does not rely on such an assumption.
\end{remark}
\begin{remark}
In general, the optimal censoring region can only be obtained numerically. The likelihood-ratio-based structure coupled with the equality constraint established in Theorem \ref{Theorem:EqualGreaterKL} can significantly reduce the computation load, especially for the scenarios where the observations are of high-dimension. The search space of the optimal censoring strategy is reduced from all admissible strategies that satisfy the energy constraint defined in \eqref{Eqn:energyconstraint} to the very special class, for which we only need to determine two parameters: the upper and lower bounds of the likelihood ratio of the \emph{no send} region.
\end{remark}

%
%
%
%

\section{Numerical Examples} \label{Section:Numerical Example}

In this section, by simulations, we show that the censoring strategy has better trade-off curves between the detection performance and the energy constraint than a random policy and the DE-CuSum algorithm proposed in \cite{banerjee2012data}.
We consider the problem of mean shift detection in Gaussian noise. Specifically, we assume that before the change event happens, the observations are i.i.d. and have the distribution
$f_0\thicksim\mathcal{N}(0,1),$
whereas the observations are i.i.d. with the post-change distribution
$f_1\thicksim\mathcal{N}(1,1).$



\emph{ Example 1.}  We compare the asymptotic detection performance, when the ARL goes to infinity, of a random policy with the asymptotic optimal censoring strategy proposed in Theorem \ref{Theorem:AsymOptCuSum}. The random policy has the form as
\begin{align}
\theta_k = \left\{
\begin{array}{ll}
1, & \text{if }  p\leq \epsilon,\\
0, & \text{otherwise},
\end{array} \right.
\end{align}
where $p$ is a random variable with a uniform distribution: $p\thicksim \text{unif(0,1)}.$
Such random policy is very simple and one of the easiest strategies to be implemented locally at the sensor nodes.

For simulation, we keep the ARL around $6500$ and simulate the ``worst-worst case'' detection delay $\mathcal{D}^{\mu}_L(T)$ defined in \eqref{Eqn:DetectionDelayLorden} under various energy levels. Note that both the random policy and the censoring strategy are stationary, so the equalizer rule (e.g., \cite[Page 134]{poor2009quickest}) holds. To simulate the detection delay we can let the change event just happens at the very beginning, i.e., $\nu=1$. The simulation result is shown in Fig. \ref{Fig:censorrandom}, from which we can see that the censoring strategy significantly outperforms the random one. In particular, even when the available average energy units of each transmission is $0.1$, there are only about $5$ extra time slots delay compared with the case when there is no energy constraint ($\epsilon=1$).
\begin{figure}
  \centering
  \includegraphics[width=3in]{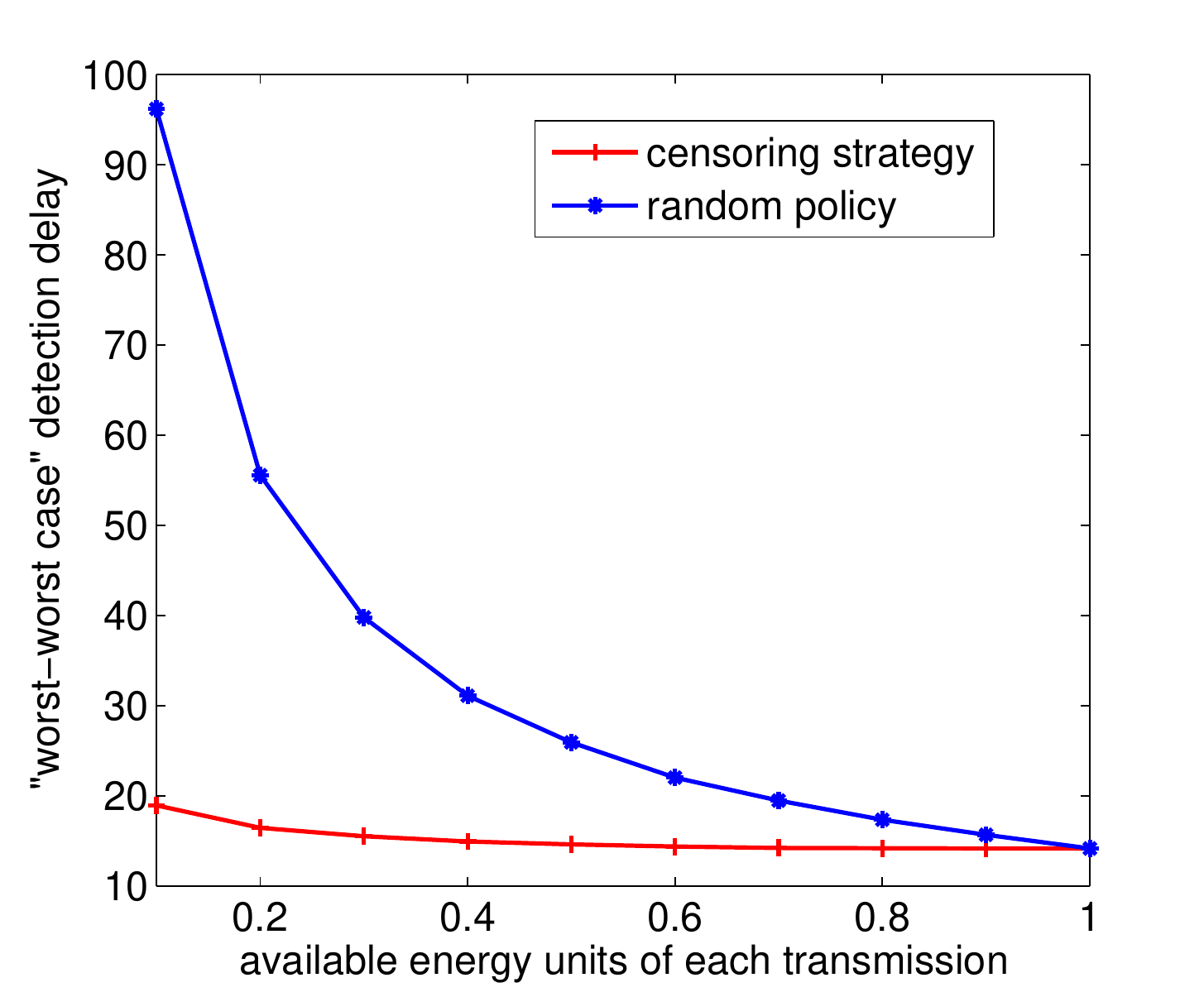}\vspace{-2mm}
  \caption{ Detection delay $\mathcal{D}^{\mu}_L(T)$ of different observation transmission scheme as a function of available average energy units of each transmission $\epsilon$. } \label{Fig:censorrandom}
  \vspace{-4mm}
\end{figure}

\emph{Example 2.}  The DE-CuSum algorithm is simulated now under the same setting as in the first example. The parameters of DE-CuSum are chosen as follows: $h=\infty$ and the deterministic incremental parameter $\mu$ (should not be confused with the notation of censoring strategy in our paper) is approximated with the following equation
\begin{align}
\mu\approx \frac{\epsilon}{1-\epsilon}\mathbb{D}(f_0||f_1). \label{Eqn:approximuDECuSum}
\end{align}
Note that for the DE-CuSum algorithm, the sequence of information available at the decision maker is correlated, so the equalizer rule does not hold any more, which causes great difficult for simulating the "worst-worst case" detection delay. The detection delay is then approximated in that we let $\nu$ be $1,2,\ldots,10$ and take the maximal one as the ``worst-worst case" detection delay. Obviously, the simulated detection delay is equal or less than the actual one associated with the DE-CuSum algorithm.
The simulation result is shown in Fig. \ref{Fig:censorDECuSum} and we can see that while the detection delay of the two schemes are approximately the same when the available energy is big enough ($\epsilon>0.5$), the censoring strategy has considerably less detection delay than the DE-CuSum algorithm when the available energy is severely limited.

\begin{figure}
  \centering
  \includegraphics[width=3in]{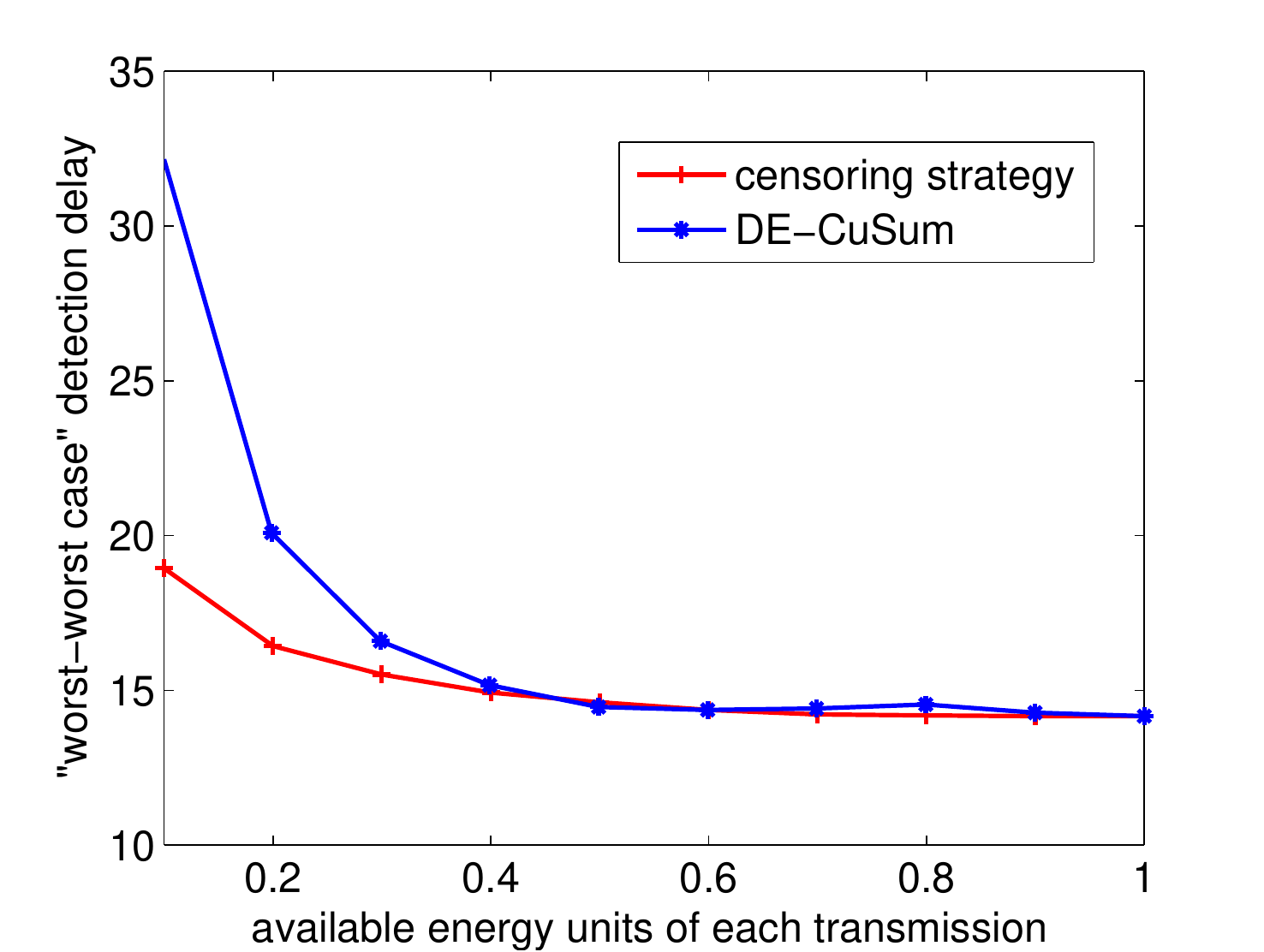}\vspace{-2mm}
  \caption{ Detection delay $\mathcal{D}^{\mu}_L(T)$ of different observation transmission scheme as a function of available average energy units of each transmission $\epsilon$. } \label{Fig:censorDECuSum}
  \vspace{-4mm}
\end{figure}

\emph{Example 3.} A typical evolution of the detection statistic $S_k$ for the CuSum algorithm with censoring strategy, random policy and the DE-CuSum algorithm is shown in Fig. \ref{Fig:TimeSeries}. The scenario where the energy is severely limited, i.e., $\epsilon=0.1$, is simulated. From \eqref{Eqn:approximuDECuSum}, to meet the energy constraint, $\mu$ is set to be $0.056$ for the DE-CuSum algorithm. To keep the ARL of the different schemes around $6500$, the threshold for censoring strategy, random policy and the DE-CuSum algorithm is set to be $690, 101$ and $98$, respectively. The change event is assumed to happen at $\nu=20$. As depicted in the figure, the censoring strategy has least detection delay. Though this evolution is just one realization of the three algorithms, it can provide insights of the reason why the censoring strategy outperforms the DE-CuSum algorithm. Note that there are deterministic increment periods for the DE-CuSum algorithm, i.e., when the detection statistic $S_k<0$, it increases with $\mu$ at each time instant regardless of the observations. Therefore, if change event happens during these periods, the DE-CuSum algorithm cannot respond to the change quickly enough. What is more, when the energy is severely limited, most of the time before the change event happens the DE-CuSum algorithm will be undergoing the deterministic increment periods, so it is very likely that the change event just happens during these periods. On the contrary, the censoring strategy is event-triggered. If the observation contains sufficient information indicating the change event (the observation lies out of the \emph{no send} region), it will be delivered to the decision maker in time. Even if the observation is not sent, the decision maker still can get the rough information of the observation (the likelihood ratio of the whole observations belong to the \emph{no send} region).

\begin{figure}
  \centering
  \includegraphics[scale=0.48]{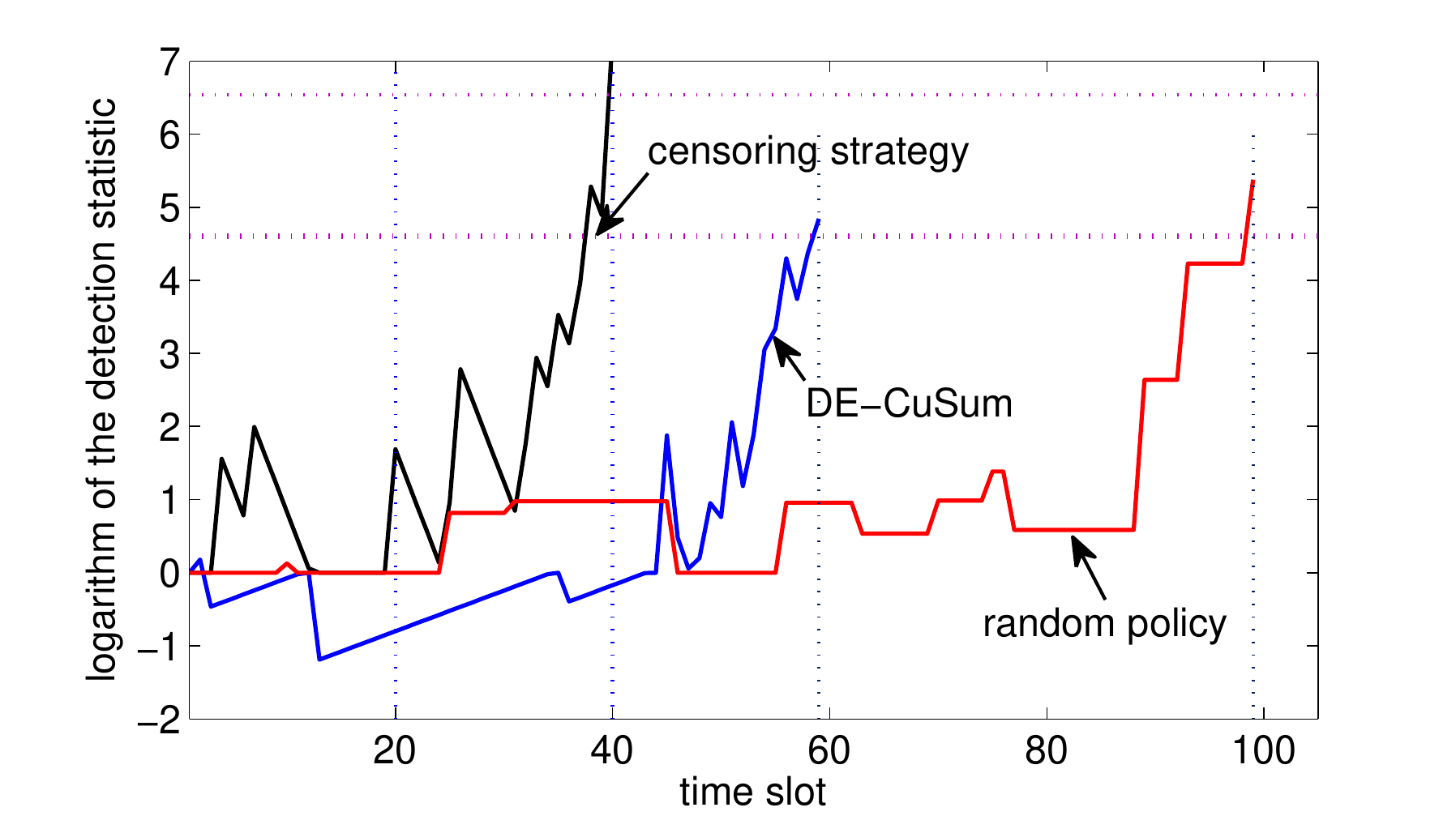}\vspace{-1mm}
  \caption{ Typical evolution of the CuSum algorithm with censoring strategy, the random policy and DE-CuSum algorithm when $\epsilon=0.1$. } \label{Fig:TimeSeries}
  \vspace{-4mm}
\end{figure}

\emph{Example 4.} The censoring strategy is now coupled with the SRP detection procedure. As in the first example, the random policy is used for comparison. Note that the SRP procedure has a randomized initial point and there is no analytic expression of the underlying distribution defined in \eqref{Eqn:QuasiStaDistri}. Hence it is impossible to simulate the SRP procedure through Monte Carlo experiments. We resort to the techniques developed in \cite{moustakides2009numerical}, i.e., solving a system of integral equations, to obtain the performance metrics: the ARL and the ``worst case" conditional average delay, $\mathcal{D}^{\mu}_P(T)$. The sample density for the integration interval $[0,A]$ is set to be $0.1$. In the scenarios, the ARL is kept to be around $1500$ by adjusting the threshold $A$. The simulation results are shown in Fig. \ref{Fig:SRPcensorrandom}. As depicted in the figure, the censoring strategy significantly outperforms the random policy. We also should note that the censoring strategy has a very good trade-off curve in itself: when $\epsilon=0.1$, the detection delay is only $3$ more time slots than that in the scenario where there is no energy constraint ($\epsilon=1$).

\begin{figure}
  \centering
  \includegraphics[width=3in]{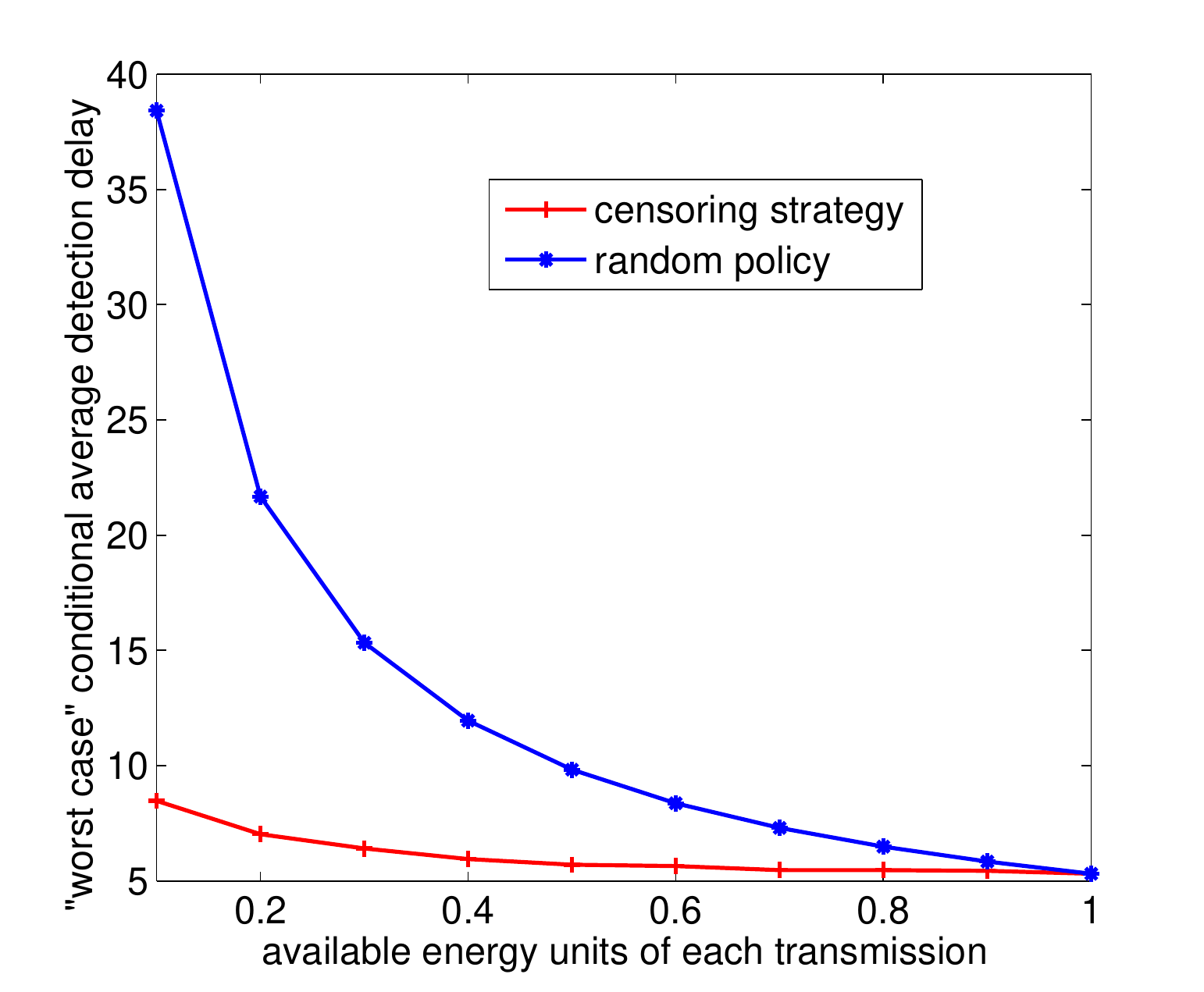}\vspace{-2mm}
  \caption{ Detection delay $\mathcal{D}^{\mu}_P(T)$ of different observation transmission scheme as a function of available average energy units of each transmission $\epsilon$. } \label{Fig:SRPcensorrandom}
  \vspace{-4mm}
\end{figure}

In the remainder, we show how we benefit from the two conditions (the two necessary properties the optimal censoring strategy has) in developing the numerical algorithm that finds the optimal censoring strategy. Note that in our case, the likelihood ratio function $L(x)=\frac{f_1(x)}{f_0(x)}$ is continuous and monotone w.r.t. $x$, which together with the likelihood-ratio-based property imply that not only the likelihood ratio of the observation but also the observation itself in the optimal censoring region is a single interval. In other words, the censoring region has the following form
\begin{align}
\theta_k = \left\{
\begin{array}{ll}
0, & \text{if }  x\in[a,b_a],\\
1, & \text{otherwise},
\end{array} \right.
\end{align}
where
\begin{align}
\int_a^{b_a}f_0(x)\mathrm{d}x=1-\epsilon.  \label{Eqn:ExampleEqualCons}
\end{align}
The numerical algorithm works as follows. As $F_0(-3.5)=1-F_1(3.5)\thickapprox0.0002$, the algorithm focus on the truncated interval $[-3.5,4.5]$ instead of $\mathbb{R}$. The real interval is discretized equally over the range with density $\delta=0.001.$ Given an energy constraint $\epsilon$, $a$ varies over these equidistant points, and the corresponding $b_a$ can be uniquely determined by the equation \eqref{Eqn:ExampleEqualCons}. For each candidate pair $(a,b_a)$, the associated K-L divergence is computed using the numerical integration technique Simpson's rule \cite{tallarida1987area}. Among all the candidate pairs, the one that has the maximal K-L divergence is the optimal one. The total time cost of our algorithm for the cases when $\epsilon=\{0.1,0.2,\ldots,1\}$ is around $0.11$ seconds. However, if there were not these two conditions, the number of intervals for the censoring region could be any positive integer and the ``size" of the censoring region could be any real value less than $\epsilon$. Therefore it would be impossible to come up with an efficient algorithm.



\section{Conclusion and Future Work} \label{Section:Conclusion}

In this paper, we studied the problem of quickest change detection in the minimax setting (both Lorden's and Pollak's formulation) in a scenario where the observations are collected by a sensor with limited energy. To deal with the energy constraint, the sensor adopts a censoring strategy, i.e., the senor only sends the observations that fall into a certain region to the decision maker. We proved that the censoring strategy that maximizes the post-censoring K-L divergence coupled with the CuSum algorithm and SRP detection procedure is asymptotic optimal, when the ARL goes to infinity, for the Lorden's and Pollak' setting, respectively. Simulation results demonstrated a considerably better performance than a random policy and the DE-CuSum algorithm. In general, to find the optimal censoring strategy can only be done numerically. We provided two properties for the optimal censoring strategy, which can be utilized to significantly reduce the computation load.

For the future work, there are multiple interesting directions: studying whether the censoring strategy that has the maximal post-censoring K-L divergence coupled with the CuSum algorithm is strictly optimal for the Lorden's problem; exploring the problem with multiple sensor nodes; and investigating time-varying censoring strategies, which may depend on the detection statistic.

\bibliographystyle{IEEETran}
\bibliography{xq_reference}

\end{document}